\documentclass[english,a4paper]{article}
\usepackage{dcolumn}
\usepackage{bm}
\usepackage{graphicx}
\usepackage{amsmath}
\usepackage{amsfonts}
\usepackage{amssymb}
\usepackage{amsthm}
\usepackage{amstext}
\usepackage{amsbsy}
\usepackage{amsopn}
\usepackage{amscd}
\usepackage{amsxtra}

\newtheorem{theorem}{Theorem}
\newtheorem{lemma}{Lemma}

\newtheorem{proposition}{Proposition}

\def\openone{\leavevmode\hbox{\small1\kern-3.3pt\normalsize1}}

\usepackage[usenames,dvipsnames]{color}

\begin{document}
\title{Geometric Origin of the Tennis Racket Effect}
\author{P. Marde\v si\'c\footnote{Institut de Math\'ematiques de Bourgogne - UMR 5584 CNRS, Universit\'e de Bourgogne- Franche Comt\'e,
	9 avenue Alain Savary, 	BP 47870, 21078 DIJON, France, pavao.mardesic@u-bourgogne.fr}, L. Van Damme, G. J. Gutierrez Guillen, D. Sugny\footnote{Laboratoire Interdisciplinaire Carnot de
Bourgogne (ICB), UMR 6303 CNRS-Universit\'e Bourgogne-Franche Comt\'e, 9 Av. A.
Savary, BP 47 870, F-21078 Dijon Cedex, France, dominique.sugny@u-bourgogne.fr}}

\maketitle


\begin{abstract}
The tennis racket effect is a geometric phenomenon which occurs in a free rotation of a three-dimensional rigid body. In a complex phase space, we show that this effect originates from a pole of a Riemann surface and can be viewed as a result of the Picard-Lefschetz formula. We prove that a perfect twist of the racket is achieved in the limit of an ideal asymmetric object. We give upper and lower bounds to the twist defect for any rigid body, which reveals the robustness of the effect. A similar approach describes the Dzhanibekov effect in which a wing nut, spinning around its central axis, suddenly makes a half-turn flip around a perpendicular axis and the Monster flip, an almost impossible skate board trick.
\end{abstract}


Consider an experiment that every tennis player has already made. The tennis racket is held by the handle and thrown in the air so that the handle makes a full turn before catching it. Assume that the two faces of the head can be distinguished. It is then observed, once the racket is caught, that the two faces have been exchanged. The racket did not perform a simple rotation around its axis, but also an extra half-turn. This twist is called the tennis racket effect (TRE). An
intuitive understanding of TRE is given in~\cite{TRE}. It is also known as Dzhanibekov's effect (DE), named after the Russian cosmonaut who made a similar experiment in 1985 with a wing nut in zero gravity~\cite{reilly,dzhanibekov}. The wing nut spins rapidly around its central axis and flips suddenly after many rotations around a perpendicular axis~\cite{dzhanibekov}.  The Monster Flip Effect (MFE) is a free style skate board trick. It consists in jumping with the skateboard and making it turn around its transverse axis with the wheels falling back to the ground. This trick is very difficult to execute since TRE predicts precisely the opposite, turning about this axis should produce a $\pi$- flip and the wheels should end up in the air. The video~\cite{monster} shows that this trick can be made with success after several attempts.

We propose in this letter to describe these phenomena. The results are established for a tennis racket and then extended to the two other systems.
The motion is modeled as a free rotation of an asymmetric rigid body, which has three different moments of inertia along its three inertia axes~\cite{Golstein50}. The axes with the smallest and largest moments of inertia are stable, while the intermediate one is unstable. It is precisely this instability which is at the origin of TRE~\cite{cushman}. A more detailed description can be obtained from Euler's equations. The three-dimensional rotation is an example of Hamiltonian integrable systems~\cite{arnold} in which the trajectories can be expressed analytically. The dynamics of the rigid body in the space-fixed frame are given by elliptic integrals of the first and third kinds, which lead to a very accurate description of TRE~\cite{cushman,MSA91,vandamme:2017}. However, this analysis does not reveal its geometric character. A geometric point of view provides valuable physical insights, in particular with respect to the robustness of the corresponding physical phenomenon. Different geometric structures have been studied recently in the context of mechanical systems with a small number of degrees of freedom. Among others, we can mention the Berry phase~\cite{phasebook}, Hamiltonian monodromy~\cite{efstathioubook,dullin:2009,co2:2004}, singular tori~\cite{sugny:2009} and the Chern number~\cite{faure} which found applications in classical and quantum physics. In this letter, we show that the geometric origin of TRE is a pole of a Riemann surface defined in a complex phase space. This effect can be interpreted as the result of the Picard-Lefschetz formula which describes the possible deformation of an integration contour in a complex space after pushing it around a singular fiber~\cite{AGV,zoladek:2006}. The geometric character of DE and MFE can also be deduced from this approach and helps understanding in which conditions they can be realized. Note that similar complex methods have been used to describe Hamiltonian monodromy~\cite{audin:2002,beukers:2002,sugny:2008}.

The position of the body-fixed frame $(x,y,z)$ with respect to the space-fixed frame $(X,Y,Z)$ defines the free rotation of a rigid body~\cite{reilly,Golstein50,arnold}. Three Euler angles $(\theta,\phi,\psi)$ characterize the relative motion of the body-fixed frame. The angle $\theta$ is the angle between the axis $z$ and the space-fixed axis $Z$. The rotation of the body about the axes $Z$ and $z$ is respectively described by the angles $\phi$ and $\psi$ (see Sup. Sec. II). The moments of inertia $I_x$, $I_y$ and $I_z$ are the elements of the diagonal inertia matrix in the body-fixed frame, with the convention $I_z<I_y<I_x$. As displayed in Fig.~\ref{fig1}, a tennis racket is a standard example of an asymmetric rigid body in which the $z$-axis is along the handle of the racket, $y$ lies in the plane of the head of the racket and $x$ is orthogonal to the head (See Sup. Sec. I). TRE consists in a $2\pi$-rotation of the body around the $y$-axis. The precession of the handle is measured by the angle $\phi$. TRE then manifests by a twist of the head about the $z$- axis, i.e. by a variation $\Delta\psi =\pi$, along a trajectory such that $\Delta\phi=2\pi$~\cite{vandamme:2017}.\\
\begin{figure}[!ht]
	\centering
	\includegraphics[width=1\linewidth]{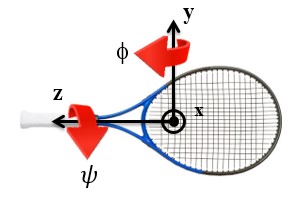}
	\caption{(Color online) A Tennis Racket with the three inertia axes $(x,y,z)$. The angles $\phi$ and $\psi$ used to define TRE describe respectively the rotation of the body around the $y$- and $z$- axes. TRE is a phenomenon in which a full turn in $\phi$- direction produces an almost perfect half-turn in $\psi$- direction.}
	\label{fig1}
\end{figure}
\emph{The Tennis Racket Effect.}-~TRE is a geometric phenomenon which does not depend on time. From Euler's equation, it can be described by the evolution of $\psi$ with respect to $\phi$~(See Sup. Sec. II):
\begin{equation}\label{eqder}
\frac{d\psi}{d\phi}=\pm\frac{\sqrt{(a+b\cos^2\psi)(c+b\cos^2\psi)}}{1-b\cos^2\psi},
\end{equation}
where we introduce the parameters $a=\frac{I_y}{I_z}-1$, $b=1-\frac{I_y}{I_x}$ and $c=\frac{2I_yH}{J^2}-1$, with the constraints $-b<c<a$, $a>0$ and $0<b<1$. $H$ and $J$ denote  respectively the rotational Hamiltonian and the angular momentum of the rigid body defined in Sup. Sec. II~\cite{Golstein50}. In the limit of a perfect asymmetric body, $I_z\ll I_y\ll I_x$, we deduce that $b\to 1$ and $a\to +\infty$. We consider only the positive values of $\frac{d\psi}{d\phi}$ defined in Eq.~\eqref{eqder}, the same analysis can be done for the negative sign. Equation~\eqref{eqder} defines a two-dimensional reduced phase space with respect to $\psi$ and $d\psi/d\phi$, as displayed in Fig.~\ref{fig5}. Note the similarity of this phase space with the one of a planar pendulum, except that two consecutive unstable fixed points are separated by $\pi$ instead of $2\pi$. The separatrix for which $c=0$ is the trajectory connecting these points~\cite{Golstein50}. We extend below the study to the complex domain and continue analytically all the functions.
\begin{figure}[!ht]
	\centering
	\includegraphics[width=1\linewidth]{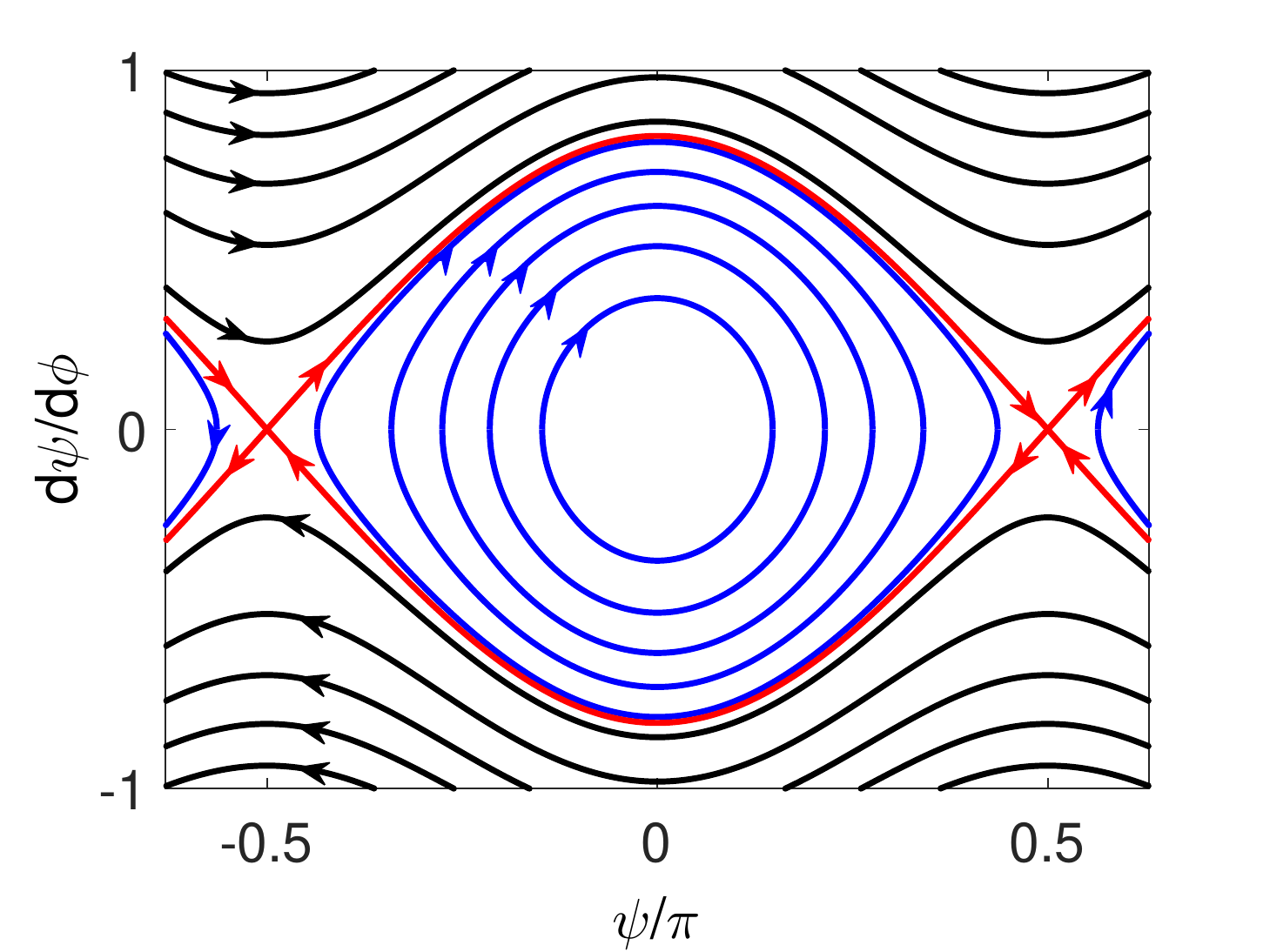}
	\caption{(Color online) Reduced phase space describing the dynamics of the rigid body in the space $(\psi,d\psi/d\phi)$. The black and
blue (dark gray) lines depict respectively the rotating and oscillating trajectories of the angular momentum. The solid red line (light gray) represents the separatrix. The parameters $a$ and $b$ are set respectively to 12 and 0.05.}
	\label{fig5}
\end{figure}

TRE is associated with a trajectory for which $\Delta\psi\simeq \pi$ when $\Delta\phi =2\pi$. We denote by $\psi_0$ and $\psi_f$ the initial and final values of the angle $\psi$. To simplify the study of TRE, we consider a symmetric configuration for which $\psi_0=-\frac{\pi}{2}+\epsilon$ and $\psi_f=\frac{\pi}{2}-\epsilon$. A perfect TRE is thus achieved in the limit $\epsilon \to 0$. Note that this symmetry hypothesis is not restrictive as shown numerically in Sup. Sec.~VI. Using Eq.~\eqref{eqder}, we obtain that the variation of $\phi$ is given by:
\begin{equation}\label{eqdeltaphi}
\Delta\phi=\int_{-\frac{\pi}{2}+\epsilon}^{\frac{\pi}{2}-\epsilon}\frac{1-b\cos^2\psi}{\sqrt{(a+b\cos^2\psi)(c+b\cos^2\psi)}}d\psi.
\end{equation}
For oscillating trajectories, the condition $c+b\cos^2\psi\geq 0$ leads to $\sin^2\epsilon\geq |\frac{c}{b}|$. From the parity of the integral and the change of variables $x=\cos^2\psi$, $\Delta\phi$ can be expressed as an incomplete elliptic integral, $\Delta\phi(\epsilon)=\int_{\sin^2\epsilon}^1\omega$, with
\begin{equation}\label{omega}
\omega=\frac{1}{b}\frac{1-bx}{\sqrt{x(x-\beta)(1-x)(x-\alpha)}}dx,
\end{equation}
where $\alpha=-\frac{a}{b}$ and $\beta=-\frac{c}{b}$. As explained in Sup. Sec.~III, we introduce a function $M$ defined by $M(u_0)=\frac{2\ln (1+\sqrt{2})}{\sqrt{1-u_0}}+2\ln (2)$ for $u_0\in ]0,1[$, and $m=M(\frac{1}{2})\simeq 3.879$. A precise description of TRE is given by Theorem~\ref{th1}, which is the main result of this study. Note that the statement is true slightly more generally for any value $u_0\in ]0,1[$, by replacing everywhere $m$ by $M(u_0)$. We put $u_0=1/2$ in Th.~\ref{th1} for simplicity.
\begin{theorem}\label{th1}
For all $c$ such that:
$$
|c|<b\exp(-2\pi\sqrt{ab}-m),
$$
for $ab$ large enough, the equation
$$
\Delta\phi_{a,b,c}(\epsilon)=2\pi
$$
has a unique solution $\epsilon_S(a,b,c)$ which verifies:
\begin{equation}\label{eqth}
\arcsin[\sqrt{|\frac{c}{b}|}]< \epsilon_S< \arcsin[\exp(-\pi\sqrt{ab}-\frac{m}{2})].
\end{equation}
This leads to:
$$
\lim_{ab\mapsto +\infty}\epsilon_S(a,b,c)=0.
$$
\end{theorem}
Several questions about its existence, uniqueness and robustness are raised by the observation of TRE, all find a rigorous answer in Th.~\ref{th1}. A first fundamental comment concerns a perfect TRE which occurs only in the limit of a very asymmetric body. Such limits are common enough in physics to reveal
specific phenomena. An example is given by the adiabatic evolution in mechanics~\cite{arnold} which is also based mathematically on an asymptotic analysis. The main statement of Th.~\ref{th1} describes the asymptotic behavior of the twist defect which approximately evolves as $\epsilon\simeq e^{-\sqrt{ab}\frac{\Delta\phi}{2}}$ for a sufficiently asymmetric body (with $ab\gg 1$). This exponential evolution is connected to the instability of the fixed points and to the presence of a pole in a complex phase space. The existence of a unique symmetric configuration realizing TRE follows from this asymptotic analysis. The corresponding trajectory is closer and closer to the separatrix for more asymmetric body (i.e. $c$ goes to 0). Theorem~\ref{th1} also establishes the robustness of TRE with respect to the shape of the body. Lower and upper bounds to the twist defect are given by Eq.~\eqref{eqth} as a function of the different parameters.

These results have a geometric origin in the complex domain. We study the solution $\epsilon$ of $\Delta\phi_{a,b,c}(\epsilon)=2\pi$, where $\Delta\phi=\Delta\phi_{a,b,c}$ is given by Eq.~\eqref{eqdeltaphi}. The origin of TRE is revealed by a complexification of the problem in which $\Delta\phi$ can be interpreted as an Abelian integral over the Riemann surface of the form $\omega$~\cite{zoladek:2006}. As displayed in Fig.~\ref{fig6}, this surface has two sheets with four branch points in $x=0$, $1$, $\beta$ and $\alpha$. Branch cuts are introduced to define a single-valued function. In the limit $c\to 0$, the two branch points $x=0$ and $x=\beta$ coincide, leading to a pole whose integral is the logarithmic function. For large values of $a$, note that there is no confluence of the branch point $x=\alpha$ with $x=\beta$ or 0.
\begin{figure}[!ht]
	\centering
	\includegraphics[width=1\linewidth]{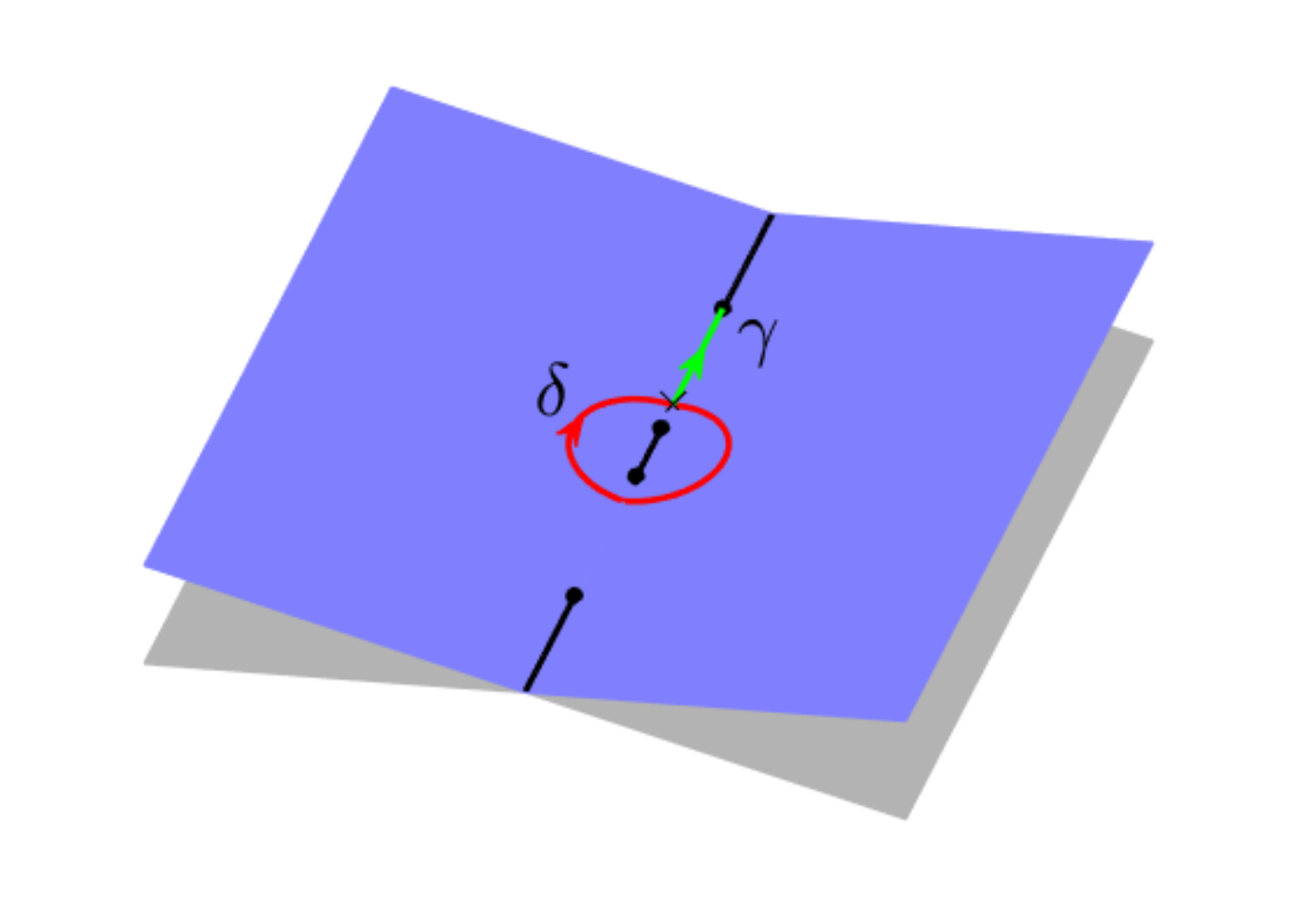}
\includegraphics[width=1\linewidth]{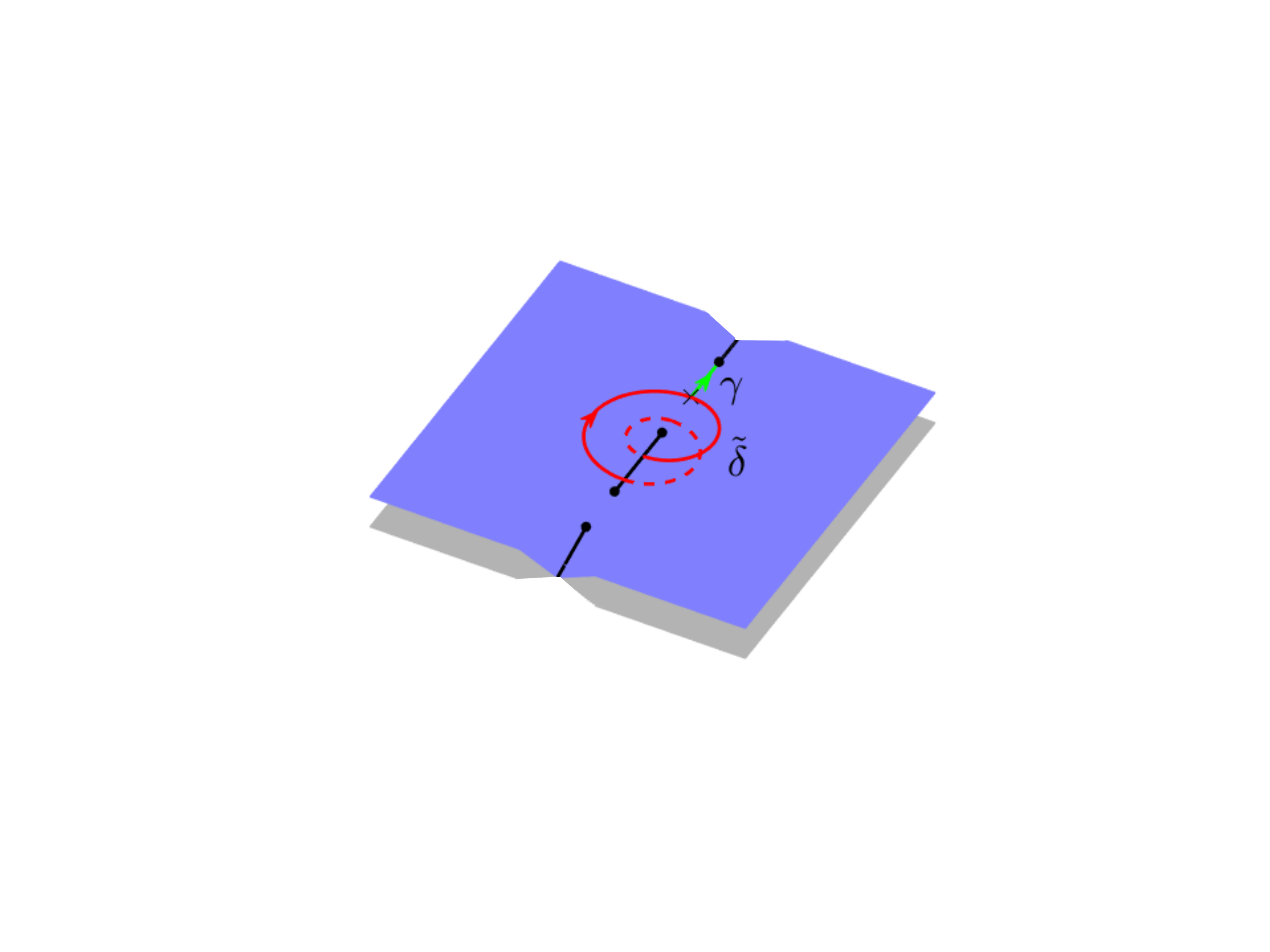}
	\caption{(Color online) Riemann surface of the form $\omega$ with the four branch points (black dots) in $x=\alpha$, $\beta$, $0$ and $1$ (from bottom to top). When $c\to 0$, the two points $x=\beta$ and $x=0$ coincide and give birth to a pole. The top and bottom panels represent the cases where the TRE can or cannot be observed. The solid straight lines represent the branch cuts of the surface. The cycles $\delta$ and $\tilde{\delta}$ are depicted by solid red (dark gray) lines. The form $\omega$ is integrated along the path $\gamma$ between the point $u$ (black cross) and the ramification point $x=1$ (green or light gray solid line). }
	\label{fig6}
\end{figure}

Let $F$ be the function defined by:
$$
F_{a,b,c}(u)=\int_{u}^{1}\omega=\int_\gamma\omega,
$$
where $\gamma$ is the integration path with $0<u<1$. We have $\Delta\phi_{a,b,c}(\epsilon)=F_{a,b,c}(\sin^2\epsilon)$. The multi-valued character of $F_{a,b,c}$ is different for $u<|\beta|$ and $u>|\beta|$. In the case $|\beta|< u <1$, we consider in the upper sheet of the Riemann surface the  cycle $\delta$ passing by $x=u$ and encircling the two branch points $\beta$ and $0$, as displayed in Fig.~\ref{fig6}. By the Picard-Lefschetz formula~\cite{AGV,zoladek:2006}, the integration contour $\gamma$ is deformed to itself plus $\delta$ when the point $x=u$ performs a loop along $\delta$. The integral $\int_\delta\omega$ adds to $F_{a,b,c}$, which reveals the multi-valued character of $F_{a,b,c}$ as a complex function. A single-valued function can be obtained by adding  a convenient multiple of $\ln u=-\int_u^1\frac{dx}{x}$, the factor being given by $\frac{1}{2\pi i}\int_\delta\omega$. In the limit $c\to 0$, $\omega$ has a pole in $x=0$ and this integral can be computed from a residue formula.

We present a heuristic proof of Th.~\ref{th1}, while a rigorous demonstration is provided in Sup. Sec.~III. We consider a simplified version of the problem where only two branch points are accounted for. We have:
$$
\int_{u}^1\frac{1}{\sqrt{x(x-\beta)}}dx=2\ln(\sqrt{x-\beta}+\sqrt{x})\big|_{x=u}^1.
$$
Using the pole at infinity, we deduce that $\frac{1}{2\pi i}\int_\delta \frac{dx}{x}=1$ and:
$$
\int_{u}^1[\frac{1}{\sqrt{x(x-\beta)}}-\frac{1}{x}]dx=2\ln\big(\frac{1+\sqrt{1-\beta}}{1+\sqrt{1-\frac{\beta}{u}}}\big),
$$
which is a well-defined and bounded function of $u$ for $|\beta|< u< 1$. As shown in Sup. Sec.~III, this argument can be generalized to $F_{a,b,c}$ which can be expressed as:
\begin{equation}\label{eqF}
F_{a,b,c}(u)=\frac{1}{\sqrt{ab}}h_{a,b,c}(u)-\frac{1}{\sqrt{ab}} \ln u,
\end{equation}
where $h_{a,b,c}$ is an analytic and bounded function in $]|\beta|,u_0[$ with $0<u_0<1$. The bound of $h_{a,b,c}$ is the function $M$ introduced in Th.~\ref{th1}. For $ab$ large enough, the equation $F_{a,b,c}(u)=2\pi$ has a unique solution which proves Th.~\ref{th1}. In the second region in which $u<|\beta|$, the geometric situation is completely different as can be seen in Fig.~\ref{fig6}. The cycle $\tilde{\delta}$ encircles only the branch point $x=0$ and no pole occurs when $c\to 0$. Turning twice around $x=0$ to get a closed path, we obtain $\int_{\tilde{\delta}}\omega =0$. This result stems from integrating the complex function $x\mapsto \frac{1}{\sqrt{x}}$ along $\tilde{\delta}$. The function $F_{a,b,c}$ is bounded with no logarithmic divergence. No information is gained about the existence, the uniqueness and the value of $\epsilon$, i.e. the possibility to realize TRE.

\emph{The Dzhanibekov effect.}-~A similar analysis can be used to describe DE~\cite{dzhanibekov}. As represented in Sup. Sec. I, the $z$- and $x$- inertia axes of this rigid body are respectively along the wings and orthogonal to the wings, while the $y$- one corresponds to the central axis of the rotation. The video~\cite{dzhanibekov} clearly shows that the motion of the wing nut is first guided by a screw which induces an almost perfect rotation around the central axis. In terms of Euler's angles, this leads to a very large angular velocity $\dot{\phi}$ and a speed $\dot{\psi}$ approximatively equal to 0 (i.e. $\frac{d\psi}{d\phi}\simeq 0$). Since the device generating the rotation of the rigid body blocks the flip motion, the angle $\psi$ is initially of the order of $\pm \frac{\pi}{2}$. We deduce that the initial point of the dynamics is very close to one of the unstable fixed points represented in Fig.~\ref{fig5}, with a parameter $c\simeq 0$. Using Eq.~\eqref{eqder}, DE is described by:
$$
\Delta\phi =\int_{-\frac{\pi}{2}}^{\frac{\pi}{2}}\frac{1-b\cos^2\psi}{\sqrt{(a+b\cos^2\psi)(c+b\cos^2\psi)}}d\psi,
$$
with $c>0$, where $\Delta\phi$ represents the angle increment before the flip of the system. We assume that the wing nut performs a perfect twist for which $\psi$ goes from $-\frac{\pi}{2}$ to $-\frac{\pi}{2}$. We show in Sup. Sec.~IV that:
\begin{equation}\label{eqdzh}
\Delta\phi= \frac{1}{\sqrt{ab}}[h_{a,b}(c)-\ln(c)],
\end{equation}
where $h_{a,b}$ is a bounded function when $c\to 0$. In this limit, the logarithmic divergence of $\Delta \phi$ occurs with the confluence of the two branch points in $x=\beta$ and $x=0$, which gives a pole as in TRE. Consequently, the speed $d\phi/d\psi$ increases tremendously in the neighborhood of this point. Note that the parameter $c$ for DE plays the same role as $\epsilon$ for TRE as can be seen in Eq.~\eqref{eqF} and \eqref{eqdzh}. DE with many rotations around the intermediate axis can be observed for a sufficient small positive value of $c$. We stress that the number of turns does not need to be complete.

\emph{The Monster Flip.}-~This approach can be used for a skate board where the $z$- and $y$- inertia axes are respectively orthogonal and parallel to the wheel axis, while the $x$- axis is orthogonal to the board (see Sup. Sec. I). MFE corresponds to a complete turn around the transverse axis together with a small variation of $\psi$. It can be realized in a neighborhood of the unstable point where $\frac{d\psi}{d\phi}=0$ (i.e. $\frac{d\phi}{d\psi}=\infty$). We search for a solution $\epsilon$ close to zero of $\tilde{\Delta}\phi(\epsilon)=2\pi$ where
\begin{equation}\label{monstereq}
\tilde{\Delta}\phi(\epsilon)=2\int_{\psi_i}^{\frac{\pi}{2}+\epsilon}\frac{1-b\cos^2\psi}{\sqrt{a+b\cos^2\psi}\sqrt{c+b\cos^2\psi}}d\psi,
\end{equation}
with $\psi_i=\pi/2$ and $\psi_i=\pi/2+\arcsin[\sqrt{|\beta|}]$ for rotating and oscillating trajectories, respectively.
As in TRE, we get $\tilde{\Delta}\phi(\epsilon)=\int_{\cos^2\psi_i}^{\sin^2\epsilon}\omega$, where $\omega$ is defined by Eq.~\eqref{omega}. Introducing $\tilde{F}_{a,b,c}(u)=\int_{\cos^2\psi_i}^{u}\omega$, it can be shown in the region $|\beta|<u<1$ that (see Sup. Sec.~V):
$$
\tilde F_{a,b,c}(u)=\frac{1}{\sqrt{ab}}\tilde h_{a,b,c}(u)+\frac{1}{\sqrt{ab}}\ln(u),
$$
where $\tilde{h}_{a,b,c}$ is a bounded and single-valued function. Note the change of sign in front of the logarithmic term with respect to Eq.~\eqref{eqF}. The solution of $\tilde{\Delta}\phi_{a,b,c}=\tilde F_{a,b,c}(u)$ can be approximated as $\epsilon \simeq\frac{\sqrt{|\beta|}}{2}e^{\pi \sqrt{ab}}$.
The accuracy of this approximation is shown numerically in Sup. Sec.~VI. For a body with $ab\geq 1$, MFE can be observed only in a neighborhood of the separatrix where $|\beta|\ll 1$. The rotation of the skate board around its transverse axis is constrained by the condition $\epsilon\geq \sqrt{|\beta|}$. This result quantifies the difficulty of performing MFE. For an angle $\epsilon$ of 30 degrees, this leads for a standard skate board to $c\simeq 10^{-3}$, while the maximum value of $c$ is of the order of 10. Finally, as illustrated in Sup. Sec.~V, MFE cannot be realized in the second region $u<|\beta|$.

\emph{Conclusion.}- TRE originates from a pole of a Riemann surface and a perfect twist of the head of the racket occurs in the limit of an ideal asymmetric body. Different properties such as the robustness of the effect have been derived from this geometric analysis. As a byproduct, we have described DE and established why the MFE is so difficult to perform. This study paves the way for the analysis of other classical integrable systems and strongly suggests the importance of complex geometry beyond the cases studied in this paper. An intriguing question is to transpose this effect to the quantum world. Different molecular systems could show traces of
this effect~\cite{RMP,rotationnumber}. Another field of applications is the control of quantum systems by external electromagnetic fields~\cite{glaser15} using, e.g., the analogy between Bloch and Euler equations~\cite{QTRE}.\\
\noindent\textbf{Acknowledgment}\\
This work was supported by the EUR-EIPHI Graduate School (Grant No. 17-EURE-0002)

\newpage

{\center\textbf{\Large{Supplemental material:\\ Geometric Origin of the Tennis Racket Effect}}}
\vspace{1cm}

This supplementary material gives a theoretical description of the Tennis Racket Eﬀect (TRE) and presents numerical simulations for different rigid bodies. The twist of TRE is schematically represented in Fig.~\ref{figS0}.

This work is organized
as follows. Section~\ref{sec1} describes the model system and the different set of numerical parameters. Standard results of rotational dynamics are recalled in Sec.~\ref{sec2}. The Euler angles used in this study are described. Using Euler's equations, we show how Eq.~(1) of the main text can be derived. Section~\ref{sec3} focuses on the proof of Theorem~I of the main text. A similar approach is applied in Sec.~\ref{secdzh} and Sec.~\ref{secMF} to describe respectively the Dzhanibekov effect and the Monster Flip effect.
Numerical results are presented in Sec.~\ref{secnum}.
\begin{figure}[!ht]
	\centering
	\includegraphics[width=0.7\linewidth]{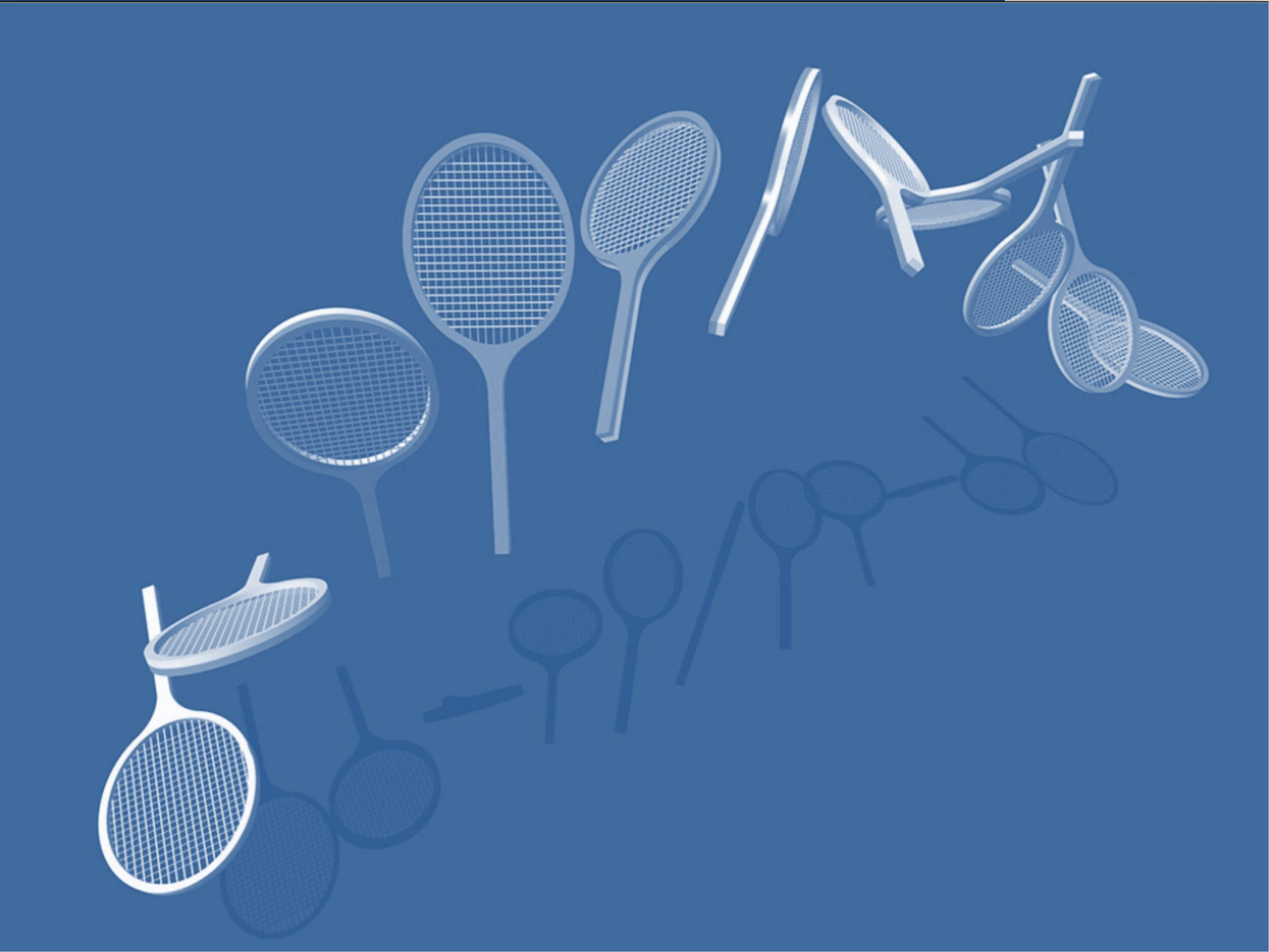}
	\caption{(Color online) Illustration of the Tennis Racket Effect: The head of the racket performs a $\pi$- flip when the handle makes a $2\pi$- rotation.}
	\label{figS0}
\end{figure}

\section{The model system}\label{sec1}
This paragraph gives some details about the different model systems used in this paper. We recall, in particular, how the moments of inertia can be estimated for different rigid bodies.

The direction of the inertia axes of a standard tennis racket is represented in Fig.~1 of the main text. Figures~\ref{figS20} and \ref{figS2} display the inertia axes of a wing nut and a skate board.


\begin{figure}[!ht]
	\centering
	\includegraphics[width=0.7\linewidth]{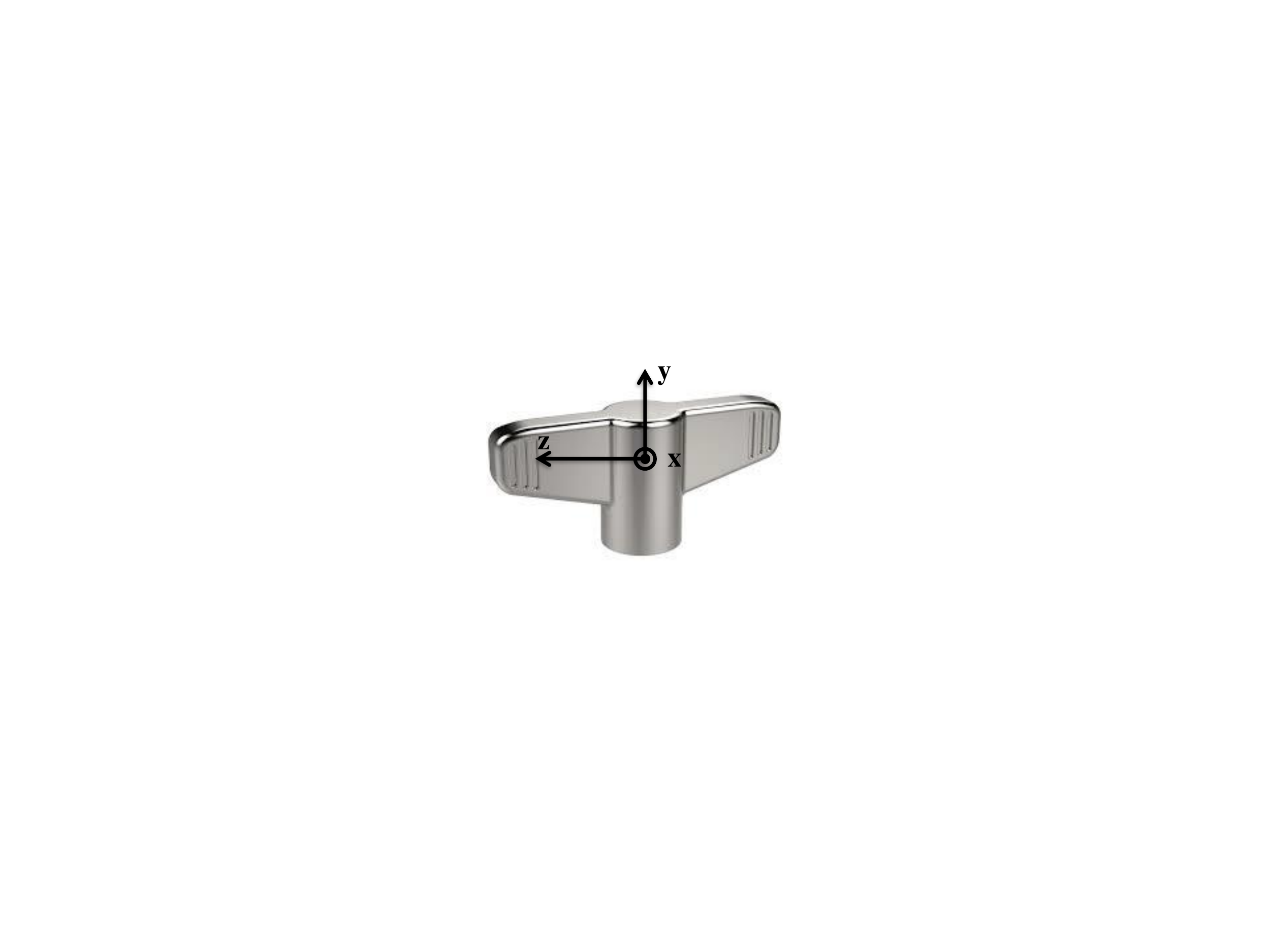}
	\caption{(Color online) A wing nut with the definition of the inertia axes ($x$, $y$, $z$). The intermediate axis is the central axis of the wing nut, while the axes with the smallest and largest moments of inertia are respectively along and orthogonal to the wings.}
	\label{figS20}
\end{figure}

\begin{figure}[!ht]
	\centering
	\includegraphics[width=0.7\linewidth]{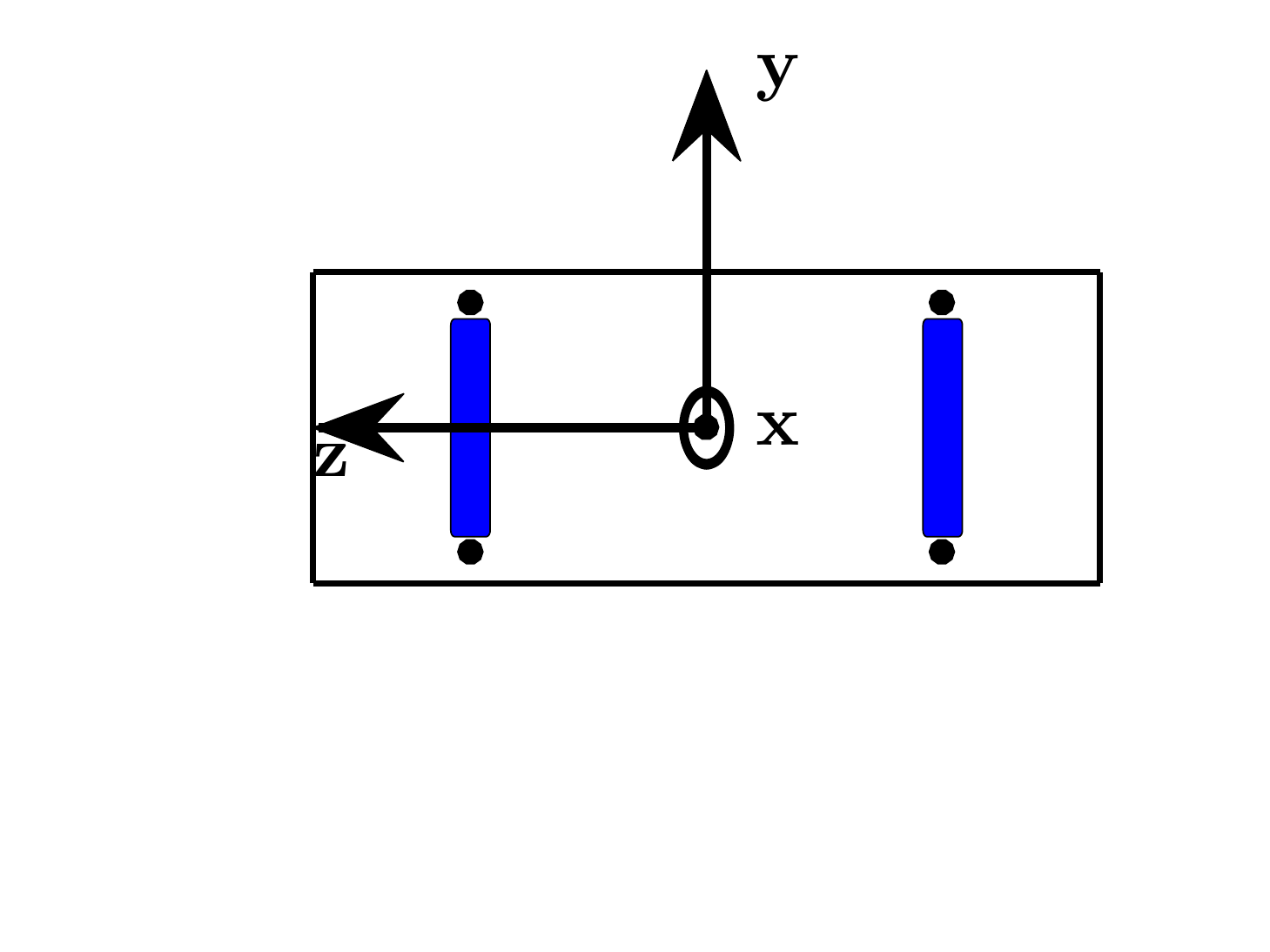}
	\caption{(Color online) Schematic representation at the scale of a skate board with the definition of the inertia axes ($x$, $y$, $z$). The big dots indicate the position of the wheels and the blue rectangles the trucks.}
	\label{figS2}
\end{figure}

The Tennis Racket Effect can also be observed with a book or a mobile phone. If the object of mass $m$ is a homogeneous rectangular cuboid of height $h$, length $L$ and width $l$, with $h<l<L$ then the moments of inertia are given by:
\begin{equation*}
\begin{cases}
I_x = \frac{m}{12}(L^2+l^2)\\
I_y = \frac{m}{12}(L^2+h^2)\\
I_z = \frac{m}{12}(h^2+l^2)\\
\end{cases}
\end{equation*}
We deduce that:
\begin{eqnarray*}
a = \frac{L^2-l^2}{h^2+l^2}\\
b = \frac{l^2-h^2}{L^2+l^2}\\
\end{eqnarray*}
If the object is almost flat, the height will be very small with respect to the other dimensions. In this case, the intermediate axis is in the plane of the object and perpendicular to the largest side. Numerical values are given in Tab.~\ref{tab1}.
\begin{table}[tb]
\caption{Numerical values of the parameters $a$ and $b$ for different objects. The book is the book of mechanics by Goldstein. The mobile phone is a Samsung JS. The moments of inertia of the wing nut and a tennis racket are given in~\cite{dzhanibekovref} and~\cite{tennispara}.\label{tab1}}
\begin{tabular}{|c|c|c|c|}
\hline
Object & $a$ & $b$ & $ab$\\
\hline
\hline
Racket & 12.54 & 0.06 & 0.75\\
\hline
Book & 1.11 & 0.31 & 0.34\\
\hline
Mobile phone & 2.97 & 0.198 & 0.59\\
\hline
Wing nut & 2.92 & 0.0972 & 0.28\\
\hline
Skate board & 8.82 & 0.078 & 0.69\\
\hline

\end{tabular}
\end{table}
The computation is more involved for a skate board since the wheels and the truck have to be accounted for. We consider the mass repartition given in Fig.~\ref{figS2} of a skate of length $L=80$~cm, width $l=20$~cm and height $h=5$~cm. The masses of the board, a wheel and a truck are estimated to be respectively of the order of $500$~g, $200$~g and $350$~g, which leads to a total mass of 2~kg. We obtain $I_x=0.123$~kg.m$^2$, $I_y=0.113$~kg.m$^2$ and $I_z=0.012$~kg.m$^2$.

\section{Euler equation for rotational motion}\label{sec2}
As mentioned in the main text, the position of the body-fixed frame $(x,y,z)$ with respect to the space-fixed frame $(X,Y,Z)$ can be described by three Euler angles $(\theta,\phi,\psi)$, which characterize the motion of the rigid body. The definition of the Euler angles is shown in Fig.~\ref{fig22}.
\begin{figure}[!ht]
	\centering
	\includegraphics[width=1\linewidth]{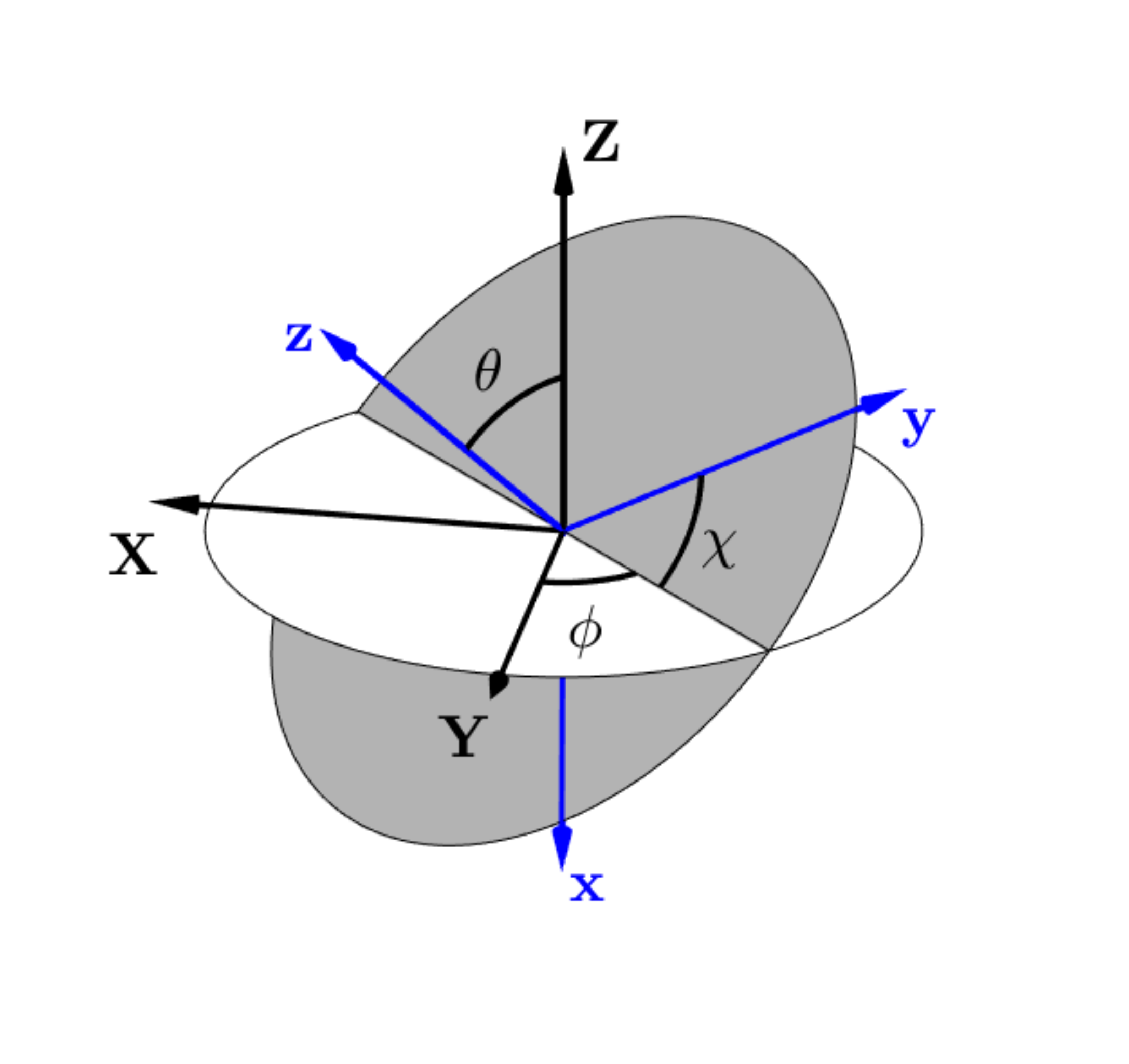}
	\caption{(Color online) Definition of the Euler angles used to describe the position
of the body-fixed frame $(x,y,z)$ with respect to the space-fixed frame $(X,Y,Z)$. The direction of the conserved angular momentum $\textbf{J}$ is also indicated.}
	\label{fig22}
\end{figure}
Rotational motion is described by integrable dynamics which have two constants of the motion, namely the angular momentum $\textbf{J}$ and the Hamiltonian $H$. The $Z$- axis of the space-fixed frame is usually chosen along the direction of $\textbf{J}$. In the body-fixed frame, the components of $\textbf{J}$ can be expressed as $J_x=-J\sin\theta\cos\psi$, $J_y=J\sin\theta\sin\psi$ and $J_z=J\cos\theta$, where $J$ is the modulus of $\textbf{J}$, while $H$ is given by $H=\frac{J_x^2}{2I_x}+\frac{J_y^2}{2I_y}+\frac{J_z^2}{2I_z}$. In the $(J_x,J_y,J_z)$- space, the energy of the intermediate axis $\frac{J^2}{2I_y}$ defines the position of the separatrix which connects the unstable fixed points and is also the boundary between the rotating and oscillating trajectories for $2I_yH>J^2$ and $2I_yH<J^2$, respectively. Using the angular velocities $\Omega_k=J_k/I_k$, $k=x,y,z$, it follows that the Euler angles satisfy the Euler differential system:
\begin{equation}\label{eqeuler}
\begin{gathered}
\begin{aligned}
& \dot{\theta}=J(\frac{1}{I_y}-\frac{1}{I_x})\sin\theta\sin\psi\cos\psi \\
& \dot{\phi}=J(\frac{\sin^2\psi}{I_y}+\frac{\cos^2\psi}{I_x}) \\
& \dot{\psi}=J(\frac{1}{I_z}-\frac{\sin^2\psi}{I_y}-\frac{\cos^2\psi}{I_x})\cos\theta
\end{aligned}
\end{gathered}
\end{equation}
We introduce the parameters $a=\frac{I_y}{I_z}-1$, $b=1-\frac{I_y}{I_x}$ and $c=\frac{2I_yH}{J^2}-1$, with the constraint $-b<c<a$. Note that $c$ measures the signed distance to the separatrix. For a standard tennis racket, we have $a=12.53$ and $b=0.063$~\cite{tennispara} while for a skate board the parameters are $a=8.82$ and $b=0.078$ (see Sec.~\ref{sec1}). Using Eq.~\eqref{eqeuler}, we can describe the dynamics and the TRE in terms of the evolution of $\psi$ with respect to $\phi$. We have:
\begin{equation}
\frac{d\psi}{d\phi}=\frac{(a+b\cos^2\psi)\cos\theta}{1-b\cos^2\psi}.
\end{equation}
From $c=a-\sin^2\theta (a+b\cos^2\psi)$, we arrive at:
\begin{equation}
\cos\theta=\pm\sqrt{\frac{c+b\cos^2\psi}{a+b\cos^2\psi}},
\end{equation}
which leads to:
\begin{equation}\label{eqder0}
\frac{d\psi}{d\phi}=\pm\frac{\sqrt{(a+b\cos^2\psi)(c+b\cos^2\psi)}}{1-b\cos^2\psi}.
\end{equation}
Equation~\eqref{eqder0} is the starting point of the analysis of the main text.

\section{Geometric proof of the Tennis Racket Effect}\label{sec3}
We show rigorously in this section the different statements about the Tennis Racket Effect.
\subsection{Analysis of the region $|\beta|< |u|< 1$}
We study in this paragraph the function $F_{a,b,c}$ defined in Eq.~(5) of the main text. From the geometric analysis in the complex space, the Picard-Lefschetz formula states that the function $F_{a,b,c}$ can be expressed in the complex domain $\mathcal{A}=\{u\in\mathbb{C}: |\beta|<|u|<1\}$ as:
$$
F_{a,b,c}(u)=g_{a,b,c}(u)-k_{a,b,c}(u)\ln u,
$$
where $g$ and $k$ are two holomorphic functions on the annulus $\mathcal{A}$, which are uniformly bounded and continuous on the closed annulus $\bar{\mathcal{A}}$. Note that $k_{a,b,c}$, which is given by $\frac{1}{2\pi i}\int_\delta \omega$, can be determined in the limit $c\to 0$ from a residue computation. However, in this example, a better upper bound and a precise expression can be derived respectively for $g_{a,b,c}$ and $k_{a,b,c}$ by considering real integrals.\\
\textbf{Analysis of $F_{a,b,c}$ in the real case:}\\
The function $F_{a,b,c}$ can be expressed as:
$$
F_{a,b,c}(u)=\frac{1}{\sqrt{ab}}(h_1(u)+h_2(u))-\frac{\ln(u)}{\sqrt{ab}},
$$
with
$$
h_1(u)=\int_{u}^1(\frac{1}{\sqrt{x(x-\beta)}}-\frac{1}{x})\frac{1-bx}{\sqrt{(1-x)(1-\frac{x}{\alpha})}}dx,
$$
and
$$
h_2(u)=\int_u^1\frac{dx}{x}[\frac{1-bx}{\sqrt{(1-x)(1-x/\alpha)}}-1].
$$
We first determine a bound in the real domain of the $h_1$- function. Let $0<u_0<1$ such that $0\leq |\beta|<u\leq u_0<1$. Since
$$
\frac{1-bx}{\sqrt{(1-x)(1+xa/b)}}\leq \frac{1}{\sqrt{1-u_0}},
$$
for $x\in ]0,u_0]$, we deduce that:
$$
|h_1(u)|\leq \frac{1}{\sqrt{1-u_0}}\int_u^1 |\frac{1}{\sqrt{x(x-\beta)}}-\frac{1}{x}|dx.
$$
which gives
$$
|h_1(u)|\leq \frac{1}{\sqrt{1-u_0}}|\int_u^1 \big(\frac{1}{\sqrt{x(x-\beta)}}-\frac{1}{x}\big) dx |.
$$
because the sign of the integrand does not change in $]|\beta|,u_0]$. As in the simplified case of the main text, we use the fact that:
$$
\int_u^1\big( \frac{1}{\sqrt{x(x-\beta)}}-\frac{1}{x}\big) dx=2\ln \big(\frac{1+\sqrt{1-\beta}}{1+\sqrt{1-\beta/u}}\big),
$$
and we arrive at
$$
|\int_u^1 \big(\frac{1}{\sqrt{x(x-\beta)}}-\frac{1}{x}\big) dx |\leq 2\ln\big(\frac{1+\sqrt{1+|\beta|}}{1+\sqrt{1-|\beta|}}\big)\leq 2\ln (1+\sqrt{2})
$$
Finally, we have:
$$
|h_1(u)|\leq \frac{2\ln (1+\sqrt{2})}{\sqrt{1-u_0}}.
$$
In a second step, we analyze the $h_2$- function. We have:
$$
|h_2(u)|\leq \int_u^1 \frac{dx}{x}|\frac{1-bx-\sqrt{(1-x)(1-x/\alpha)}}{\sqrt{(1-x)(1-x/\alpha)}}|\leq \int_u^1\frac{dx}{x}\frac{1-\sqrt{1-x}}{\sqrt{1-x}},
$$
which is valid for $ab$ large enough. The upper bound of $h_2$ can be exactly integrated:
$$
|h_2(u)|\leq \big| [-2\ln(1+\sqrt{1-x})]_u^1\big| \leq 2\ln (2)
$$
We finally get:
$$
F_{a,b,c}(u)=\frac{h_{a,b,c}(u)}{\sqrt{ab}}-\frac{\ln u}{\sqrt{ab}},
$$
where $h_{a,b,c}=h_1+h_2$ is a bounded function
with
$$
|h_{a,b,c}(u)|\leq \frac{2\ln (1+\sqrt{2})}{\sqrt{1-u_0}}+2\ln (2),
$$
which is the bound used in the main text. Note that the bound on $h_{a,b,c}$ does not depend on $a$, $b$ and $c$ but only on a fixed parameter $u_0$ which can be chosen at will in $]0,1[$. We denote by $M$ the function defined by $M(u_0)=\frac{2\ln (1+\sqrt{2})}{\sqrt{1-u_0}}+2\ln (2)$ for $u_0\in ]0,1[$. The derivative $h'_{a,b,c}$ of $h_{a,b,c}$ is given by the corresponding integrands of $h_1$ and $h_2$, leading to:
$$
h'_{a,b,c}(u)=\frac{1}{u}-\frac{1}{\sqrt{u(u+c/b)}}\frac{1-bu}{\sqrt{(1-u)(1+ub/a)}}.
$$
\begin{proposition}\label{propTRE}
For all $u_0\in ]0,1[$, for all $c$ such that
\begin{equation}\label{condgamma}
0\leq |c|<be^{-2\pi\sqrt{ab}-M(u_0)},
\end{equation}
for $ab$ large enough, the equation
\begin{equation}\label{Fpi}
F_{a,b,c}(u)=2\pi
\end{equation}
has a unique solution $u=u_S(a,b,c)$ in $]|\frac{c}{b}|,u_0[$, which verifies:
$$
|\frac{c}{b}|< u_S< e^{-2\pi\sqrt{ab}+M(u_0)}
$$
and, in particular,
$$\lim_{ab\mapsto+\infty}u_S(a,b,c)=0.$$
\end{proposition}
\begin{proof}
Equation \eqref{Fpi} becomes:
\begin{equation}\label{TRE}
\frac{1}{\sqrt{ab}}h_{a,b,c}(u)-\frac{1}{\sqrt{ab}}\ln u=2\pi
\end{equation}
Equation~\eqref{TRE} can be expressed in terms of a fixed point problem $u=f(u)$, with
$$
f(u)=e^{-2\pi\sqrt{ab}+h_{a,b,c}(u)}.
$$
We arrive at:
\begin{equation}\label{inequality}
e^{-2\pi\sqrt{ab}-M(u_0)}<f(u)<e^{-2\pi\sqrt{ab}+M(u_0)}.
\end{equation}
We show by continuity the existence of a solution to the fixed point problem if $f(|\beta|)>|\beta|$ and $f(u_0)<u_0$. The first condition is given by Eq.~\eqref{condgamma} while the second inequality is trivially verified from Eq.~\eqref{inequality}, for $ab$ large enough. The uniqueness of the solution is verified if the function $g:~u\mapsto f(u)-u$ is strictly decreasing. We show this statement for $c\leq 0$, while for $c>0$, we prove that $g$ is increasing on $[|\beta|,u_m[$, it reaches its maximum in $u=u_m$ and is strictly decreasing on $]u_m,u_0[$.

Let us first consider the case $c\leq 0$. The function $h'_{a,b,c}$ can be bounded for $u\in ]|\beta|,u_0]$ by:
$$
h'_{a,b,c}(u)\leq \frac{1}{u}(1-\frac{1-bu}{\sqrt{(1-u)(1+ub/a)}})\leq t(u)
$$
where
$$
t(u)=\frac{1}{u}(1-(1-u)^{-1/2}).
$$
Since $\lim_{u\to 0}t(u)=-\frac{1}{2}$ and $t$ is a strictly decreasing function, we deduce that $h'(u)\leq -\frac{1}{2}$ for $u\in ]|c/b|,u_0[$. $g$ is therefore also strictly decreasing.

We then study the case $c>0$. A zero $u_m$ of $g'$ fulfills:
$$
h'_{a,b,c}(u_m)e^{h_{a,b,c}(u_m)}=e^{2\pi \sqrt{ab}},
$$
then:
\begin{equation}\label{eqder}
e^{2\pi \sqrt{ab}-M(u_0)}<h'_{a,b,c}(u_m)<e^{2\pi \sqrt{ab}+M(u_0)}
\end{equation}
For $ab$ large enough, Eq.~\eqref{eqder} shows that $u_m$ belongs to a small neighborhood of $u=|c/b|$ when $|c/b|\ll 1$. Moreover, the function $h'_{a,b,c}$ can be bounded by:
$$
h'_{a,b,c}(u)\leq r(u),
$$
where $r(u)=\frac{1}{u}-\frac{1}{\sqrt{u(u+|c/b|)}\sqrt{1-u}}$. We have:
$$
r(u)\leq 0\Leftrightarrow u^2+|c/b|u-|c/b|>0.
$$
We obtain that $r(u)\leq 0$ and $h'_{a,b,c}\leq 0$ if $u\geq \sqrt{|c/b|}$ when $|c/b|\to 0$. In the interval $[|c/b|,\sqrt{|c/b|}]$, $h'_{a,b,c}$ is equivalent for $|c/b|\ll 1$ to:
$$
h'_{a,b,c}(u)\simeq \frac{1}{u}-\frac{1}{\sqrt{u(u+|c/b|)}},
$$
which is a strictly decreasing function tending to $+\infty$ when $u$ and $c$ go to 0. We deduce that there exists a unique $u_m$ such that $g'(u_m)=0$. We finally obtain that $g(u_m)>0$ and $g'(u)<0$ in $]u_m,u_0]$, which leads to the uniqueness of the solution $u_S$.
\end{proof}
Using Proposition~\ref{propTRE}, we can deduce Theorem~1 of the main text.
\begin{proof}
The proof follows directly from Proposition~\ref{propTRE} and the relation $\Delta\phi_{a,b,c}(\epsilon)=F_{a,b,c}(\sin^2\epsilon)$, since the change of variables $u=\sin^2\epsilon$ is a bijection from $[0,\pi/2]$ to $[0,1]$.
\end{proof}

Note that Proposition~\ref{propTRE} and Theorem~1 can alternatively be proved using the fixed point theorem. For $ab$ large enough, the condition \eqref{inequality} gives that $f:]|\beta|,u_0[\to ]|\beta|,u_0[$. The above proof was used because it gives in addition an interval in which the respective fixed points $u_S$ and $\epsilon_S$ of $f$ and $f\circ\arcsin$ belong, showing thus the corresponding limits for $u_S$ and $\epsilon_S$, when $ab\to+\infty$. On the other hand, the fixed point theorem also shows the robustness of the phenomenon.
\subsection{Analysis of the region $u<|\beta|$}
We consider now the function $F_{a,b,c}$ in the region $|u|<|\beta|$. We recall that this analysis only concerns the case with $c>0$ and that the result of Eq.~(5) of the main text does not hold.
\begin{lemma}
There exists a holomorphic function $k$ defined on
$$\mathcal{D}=\{v\in\mathbb{C}: |v|<\sqrt{|\beta|}\}$$
such that
$$
F_{a,b,c}(u)=k(\sqrt{u})
$$
i.e. $F(v^2)=k(v)$.
\end{lemma}
\begin{proof}
Turning around the origin in $u$, we do not catch the cycle $\delta$ as in the TRE, but a non-closed path. Turning twice around $x=0$, we catch a closed cycle $\tilde{\delta}$ winding twice around the branch point $x=0$ only. Note that here $\int_{\tilde{\delta}}\omega=0$. The result is equivalent to integrate $x\mapsto \frac{1}{\sqrt{x}}$ on a loop winding twice around zero.
Let $k$ be $k(v)=F_{a,b,c}(v^2)$. Then, we deduce that:
$$k(ve^{2\pi i})=F_{a,b,c}(v^2e^{4\pi i})=F_{a,b,c}(v^2)+\int_{\tilde{\tilde{\delta}}}\omega=F_{a,b,c}(v^2)=k(v).$$
Moreover, $k(0)=\int_0^1\omega<\infty$ is a complete elliptic integral.
Hence, $k$ has a removable singularity at the origin and extends to a holomorphic function on $\mathcal{D}.$
\end{proof}

Equation $\Delta\phi_{a,b,c}(\epsilon)=F_{a,b,c}(\sin^2\epsilon)$ becomes
$$
h_{a,b,c}(\sin^2\epsilon)=2\pi,
$$
where $h_{a,b,c}$ is a bounded and analytic function. Note that the nature of this equation, valid in the small region $\mathcal{D}$ is completely different from Eq.~\eqref{TRE}.

\section{Analysis of the Dzhanibekov effect}\label{secdzh}
After a very large number of rotations around its central axis, the wing nut flips suddenly around a perpendicular axis. This phenomenon occurs for rotating trajectories for which $c>0$. We have:
$$
\Delta\phi =\int_{-\frac{\pi}{2}}^{\frac{\pi}{2}}\frac{1-b\cos^2\psi}{\sqrt{(a+b\cos^2\psi)(c+b\cos^2\psi)}}d\psi.
$$
With the same change of coordinates as for TRE, we arrive at:
$$
\Delta\phi = \int_0^1 \frac{1}{b}\frac{1-bx}{\sqrt{x(x-\beta)(1-x)(x-\alpha)}}dx
$$
with $\alpha=-\frac{a}{b}$ and $\beta=-\frac{c}{b}$. This integral can be viewed as an Abelian integral for a cycle connecting the two branch points 0 and 1. It starts on one sheet of the Riemann surface and ends on the other.
We now estimate the integral in the real domain. We have:
$$
\Delta\phi =\frac{1}{\sqrt{ab}}\int_0^1\frac{1-bx}{\sqrt{x(x-\beta)(1-x)(1-\frac{x}{\alpha})}}.
$$
The variation $\Delta\phi$ can be written as the sum of two terms:
$$
\Delta\phi=\frac{1}{\sqrt{ab}}[h(c)+g(c)],
$$
with
$$
g(c)=\int_0^1\frac{1}{\sqrt{x(x-\beta)}}dx.
$$
Straightforward computations lead to:
$$
g(c)=2\ln(\frac{\sqrt{b}+\sqrt{b+c}}{\sqrt{c}}).
$$
We show in a second step that the function $h$ is bounded. We have:
$$
h(c)=\frac{1}{\sqrt{ab}}\int_0^1\frac{1}{\sqrt{x(x-\beta)}}(\frac{1-bx}{\sqrt{(1-x)(1-\frac{x}{\alpha})}}-1)dx.
$$
We can derive an upper bound as follows:
$$
|h(c)|\leq \frac{1}{\sqrt{ab}}\int_0^1\frac{1}{\sqrt{x(x-\beta)}}|\frac{1-bx-\sqrt{(1-x)(1-\frac{x}{\alpha})}}{\sqrt{(1-x)(1-\frac{x}{\alpha})}}|dx
$$
We obtain:
$$
|h(c)|\leq \frac{1}{\sqrt{ab}}\int_0^1\frac{1}{x}[\frac{1-\sqrt{1-x}}{\sqrt{1-x}}]dx\leq 2\ln 2,
$$
which is valid for $ab$ large enough. It is then straightforward to derive Eq.~(7) of the main text:
\begin{equation*}
\Delta\phi= \frac{1}{\sqrt{ab}}[h_{a,b}(c)-\ln(c)].
\end{equation*}
When $c\to 0$, we can estimate the variation $\Delta\phi$. Direct computations lead to:
$$
\Delta\phi\simeq \frac{1}{\sqrt{ab}}[\ln(\frac{4b}{c})+2\ln 2].
$$
The accuracy of this approximation is investigated in Fig.~\ref{figdzh} for a standard wing nut.
\section{Analysis of the Monster Flip}\label{secMF}
We study in this section the Monster Flip for a standard skate board. As in the main text, we introduce the function $\tilde{F}_{a,b,c}=\int_{\cos^2\psi_i}^{\sin^2\epsilon}$ and search for solutions of
\begin{equation}\label{monster}
\tilde{F}_{a,b,c}(u)=2\pi, \quad 0\leq u\leq 1,
\end{equation}
in the rotating case or
\begin{equation}\label{monster2}
\tilde{F}_{a,b,c}(u)=2\pi, \quad |\beta|\leq u\leq 1,
\end{equation}
for oscillating trajectories. Note that $\cos^2\psi_i$ is equal to 0 or to $|\beta|$ in the rotating and oscillating cases, respectively. Following the study used in TRE, we consider the two regions $0<u< |\beta|$ and $|\beta|<u<1$.
In the case $0<u< |\beta|$, which only concerns rotating trajectories, it can be shown that:
$$ \tilde{F}_{a,b,c}(u)=\tilde{h}_{a,b,c}(\sqrt{u}), $$
where $\tilde{h}_{a,b,c}$ is a holomorphic function vanishing at the origin. For the region ${|\beta|}<|u|<1$, we get:
$$
\tilde F_{a,b,c}(u)-\ln(u)\int_\delta\omega=\frac{1}{\sqrt{ab}}\tilde h_{a,b,c}(u),
$$
where $\tilde{h}_{a,b,c}$ is a single-valued function.

We consider now the different integrals in the real domain. Starting from the equation $\Delta\phi_{a,b,c}=\tilde F_{a,b,c}(u)$, we deduce that
\begin{equation}\label{M}
\tilde{F}_{a,b,c}(u)=e^{2\pi\sqrt{ab}+\tilde{h}_{a,b,c}(u)}.
\end{equation}
Approximate expressions of the variation $\epsilon$ can be obtained as follows. When $u\ll 1$, we have:
$$
\tilde F_{a,b,c}(u)=\frac{1}{b}\int_{\cos^2\psi_i}^u\frac{1-bx}{\sqrt{x(x-\beta)(1-x)(x-\alpha)}}dx\simeq \frac{1}{\sqrt{ab}}\int_0^u \frac{dx}{\sqrt{x(x-\beta)}},
$$
where we have replaced $x$ by 0 except in the factor $\sqrt{x(x-\beta)}$. A standard integration leads to:
$$
\tilde F_{a,b,c}(u)\simeq \frac{2}{\sqrt{ab}}\ln \big(\sqrt{1+\frac{u}{|\beta|}}+\sqrt{\frac{u}{|\beta|}}\big)
$$
The equation $\Delta\phi_{a,b,c}=2\pi=\tilde F_{a,b,c}(u)$ can then be approximated as:
$$
\sqrt{1+\frac{u}{|\beta|}}+\sqrt{\frac{u}{|\beta|}}=e^{\pi\sqrt{ab}}.
$$
In the case $0<u<|\beta|$, we have $\sqrt{1+\frac{u}{|\beta}}+\sqrt{\frac{u}{|\beta|}}\leq 1+\sqrt{2}$ and we recover the fact that the $\tilde{h}_{a,b,c}$- function is bounded. This also gives a strong constraint on the parameters $a$ and $b$:
$$
ab\leq \frac{[\ln(1+\sqrt{2})]^2}{\pi^2}
$$
The bound on the product $ab$ is of the order of 0.079 which means that this situation is not very interesting in practice since the rigid body has to be slightly asymmetric. The variation $\epsilon$ of MFE can be estimated as:
\begin{equation}\label{epsi1}
\epsilon\simeq \pi \sqrt{ac}.
\end{equation}
In the region $u>|\beta|$, a simple formula can be derived in the limit $u/|\beta|\gg 1$. A first order Taylor expansion leads to:
\begin{equation}\label{epsi2}
\epsilon\simeq \frac{\sqrt{|\beta|}}{2}e^{\pi\sqrt{ab}},
\end{equation}
which allows to estimate the bounded function $\tilde{h}_{a,b,c}$. These different approximations will be illustrated numerically in Sec.~\ref{secnum}.
\section{Numerical results}\label{secnum}
The goal of this paragraph is to illustrate numerically the different results established in this work.
Figure~\ref{fig7} gives a general overview of the twist $|\Delta \psi|$ of the head of the racket when the handle makes a $2\pi$- rotation. The twist is plotted as a function of the initial conditions $\psi_0$ and $d\psi/d\phi|_0$. In addition, this numerical result shows that the TRE and the MFE are not limited to the symmetric configuration analyzed in this study. We observe that TRE can be achieved in a large area around the separatrix. MFE occurs only in a very small band around the separatrix, which shows the difficulty to realize the Monster flip.
\begin{figure}[!ht]
	\centering
	\includegraphics[width=0.7\linewidth]{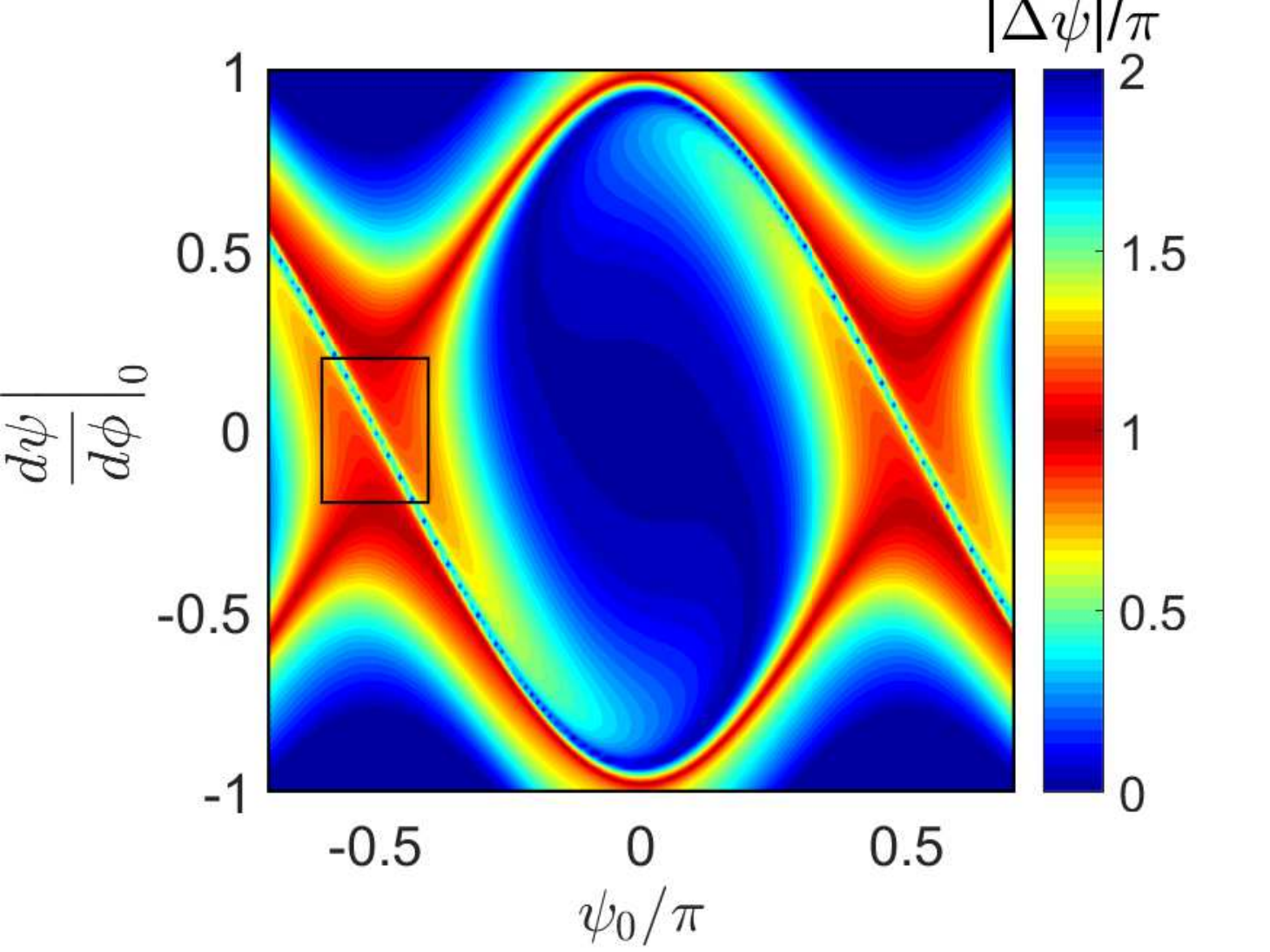}
    \includegraphics[width=0.7\linewidth]{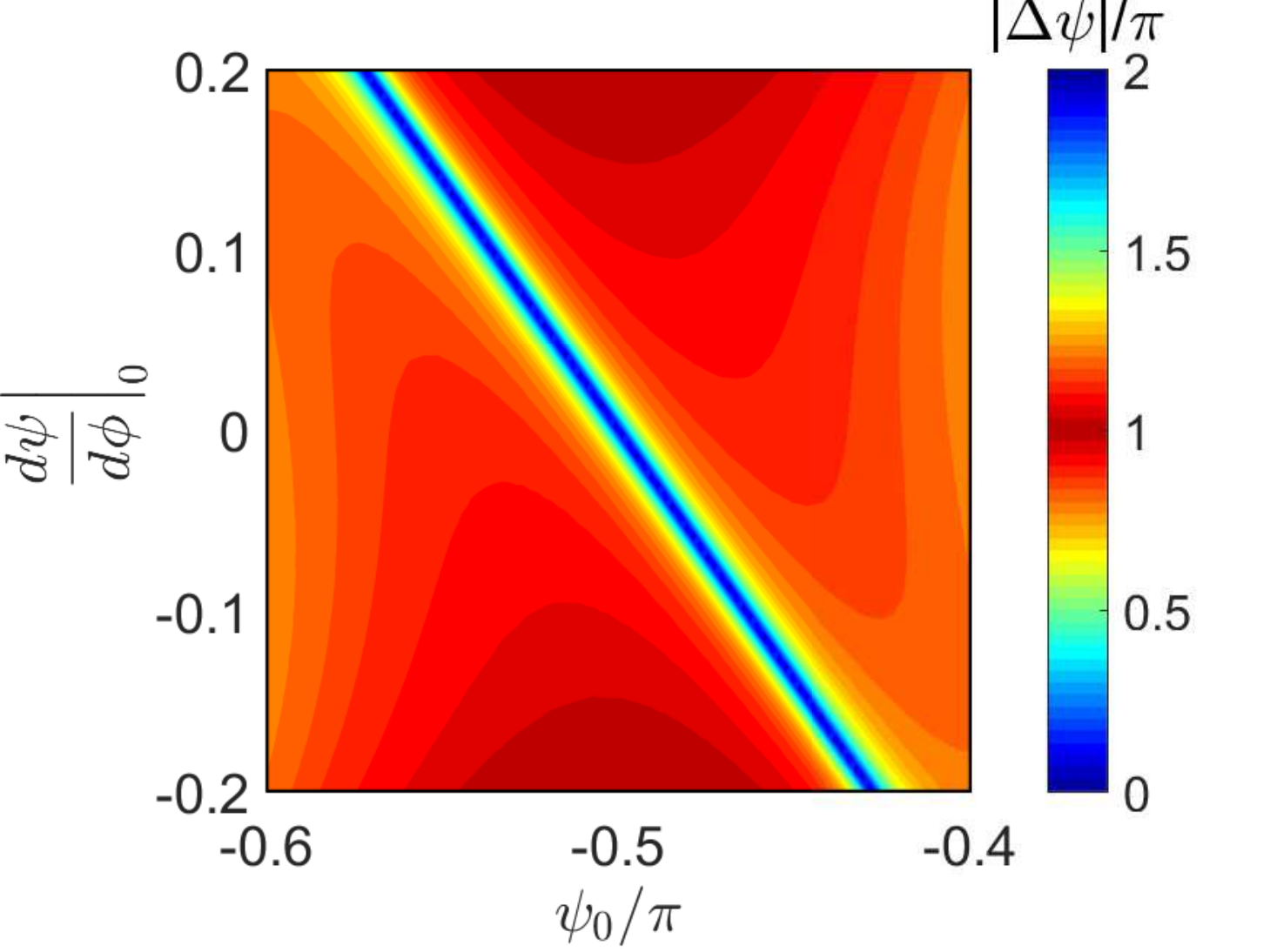}
	\caption{(Color online) (upper panel) Contour plot of the twist $|\Delta \psi|$, with $\Delta\phi=2\pi$, as a function of the initial conditions $\psi_0$ and $d\psi/d\phi|_0$. (lower panel) Zoom of the upper panel corresponding to the black rectangle.}
	\label{fig7}
\end{figure}

Figure~\ref{fig8} displays the evolution of $\epsilon$ in the TRE case. We observe that $\epsilon$ goes to zero when $a$ increases. For the chosen value of the $c$ parameter, it can be verified that $\epsilon >\epsilon_0$ for $a\leq 120$, with $\epsilon_0=\textrm{arcsin}[\sqrt{|c/b|}]$.
\begin{figure}[!ht]
	\centering
	\includegraphics[width=0.7\linewidth]{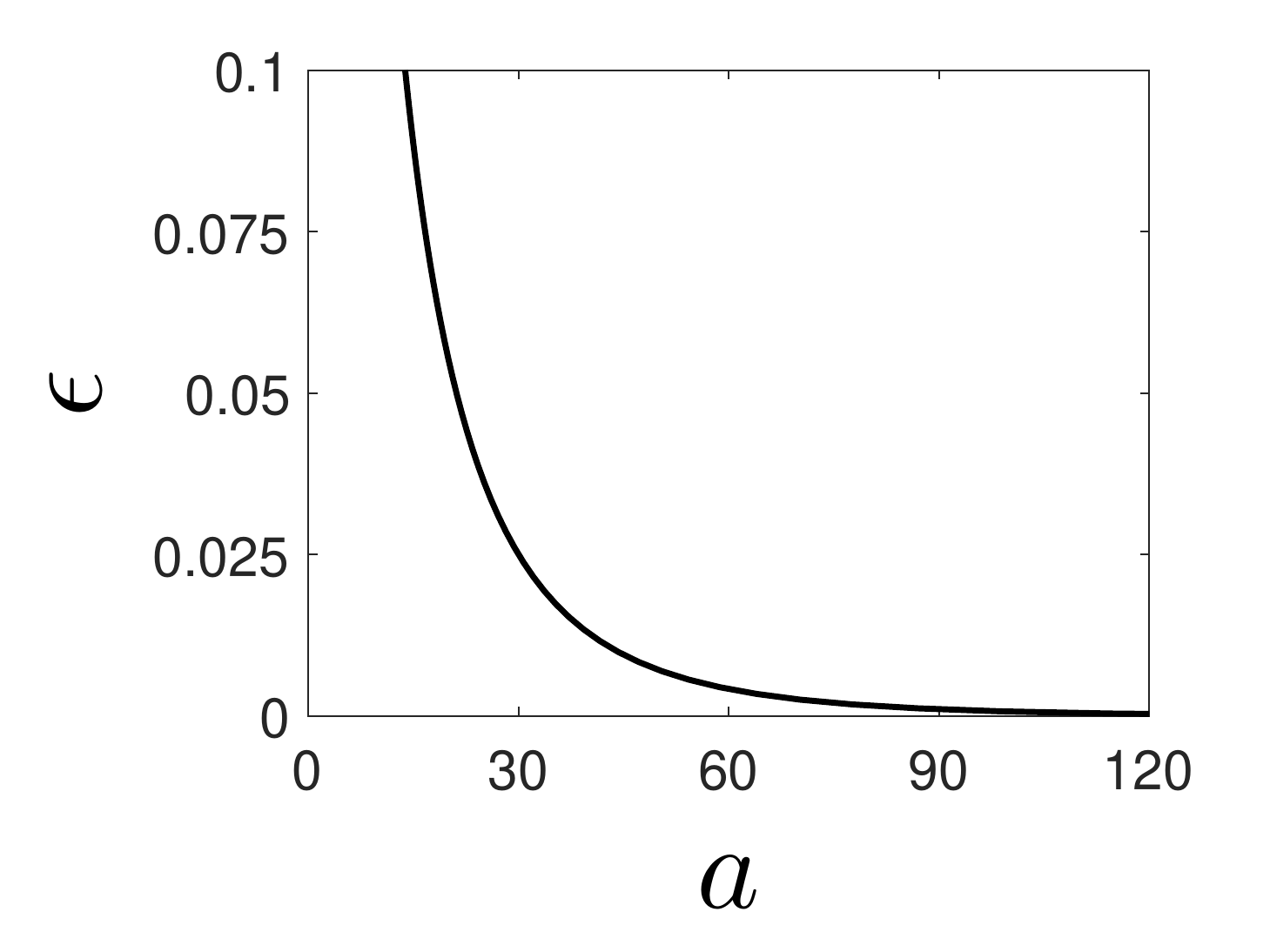}
\caption{(Color online) Evolution of $\epsilon$ as a function of $a$ in the TRE case. Parameters are set to $c=10^{-9}$ and $b=0.0629$, which leads to $\epsilon_0 = 1.26\times 10^{-4}$.}
	\label{fig8}
\end{figure}

We study in Fig.~\ref{figdzh} the evolution of $\Delta\phi$ with respect to the parameter $c$ for a standard wing nut. We observe the divergence of $\Delta\phi$ when $c$ goes to 0. Using the analysis of Sec.~\ref{secdzh}, we approximate with a good accuracy $\Delta\phi$ as follows:
\begin{equation}\label{phiapp}
\Delta\phi\simeq \frac{1}{\sqrt{ab}}[2\ln 2+\ln(\frac{4b}{c})].
\end{equation}
\begin{figure}[!ht]
	\centering
	\includegraphics[width=0.7\linewidth]{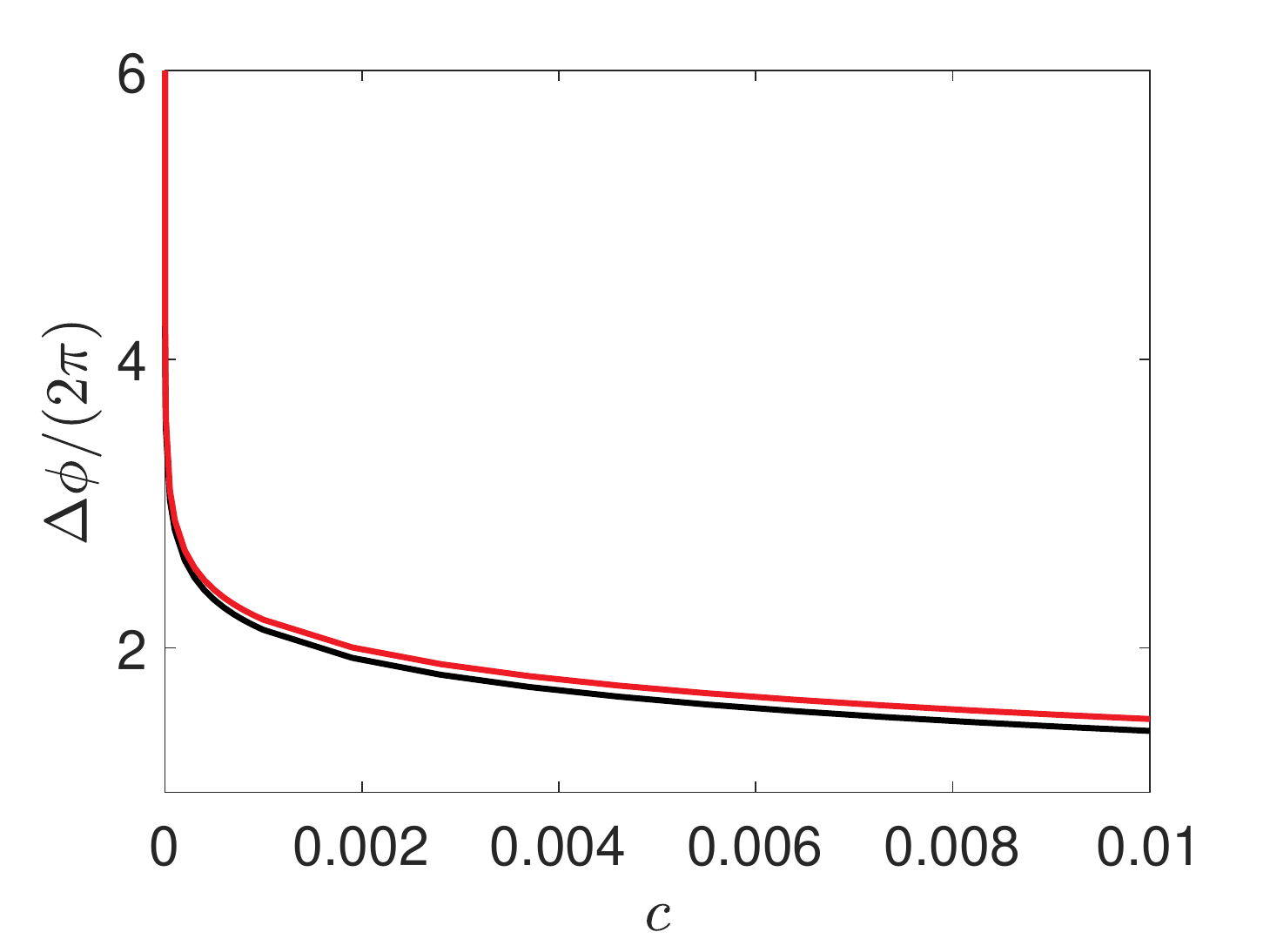}
\caption{(Color online) Evolution of the variation $\Delta\phi$ for the DE case as a function of $c$ (solid black line). The approximation of $\Delta\phi$ given in Eq.~\eqref{phiapp} is plotted in red (or dark gray). Parameters are set to $a=2.92$ and $b=0.097$.}
	\label{figdzh}
\end{figure}

The behavior of $\epsilon$ in the MFE case is represented respectively in Fig.~\ref{fig9} and \ref{fig10} in the region where $\epsilon<\epsilon_0$ and $\epsilon>\epsilon_0$. It can be seen that Eq.~(\ref{epsi1}) and (\ref{epsi2}) give a very good estimate of $\epsilon$. As could be expected, small values of $\epsilon$ are only obtained when $c$ is sufficiently small. The parameters $a$ and $b$ have been chosen so that $ab\leq \frac{[\ln(1+\sqrt{2})]^2}{\pi^2}$ in the first situation.
\begin{figure}[!ht]
	\centering
	\includegraphics[width=0.7\linewidth]{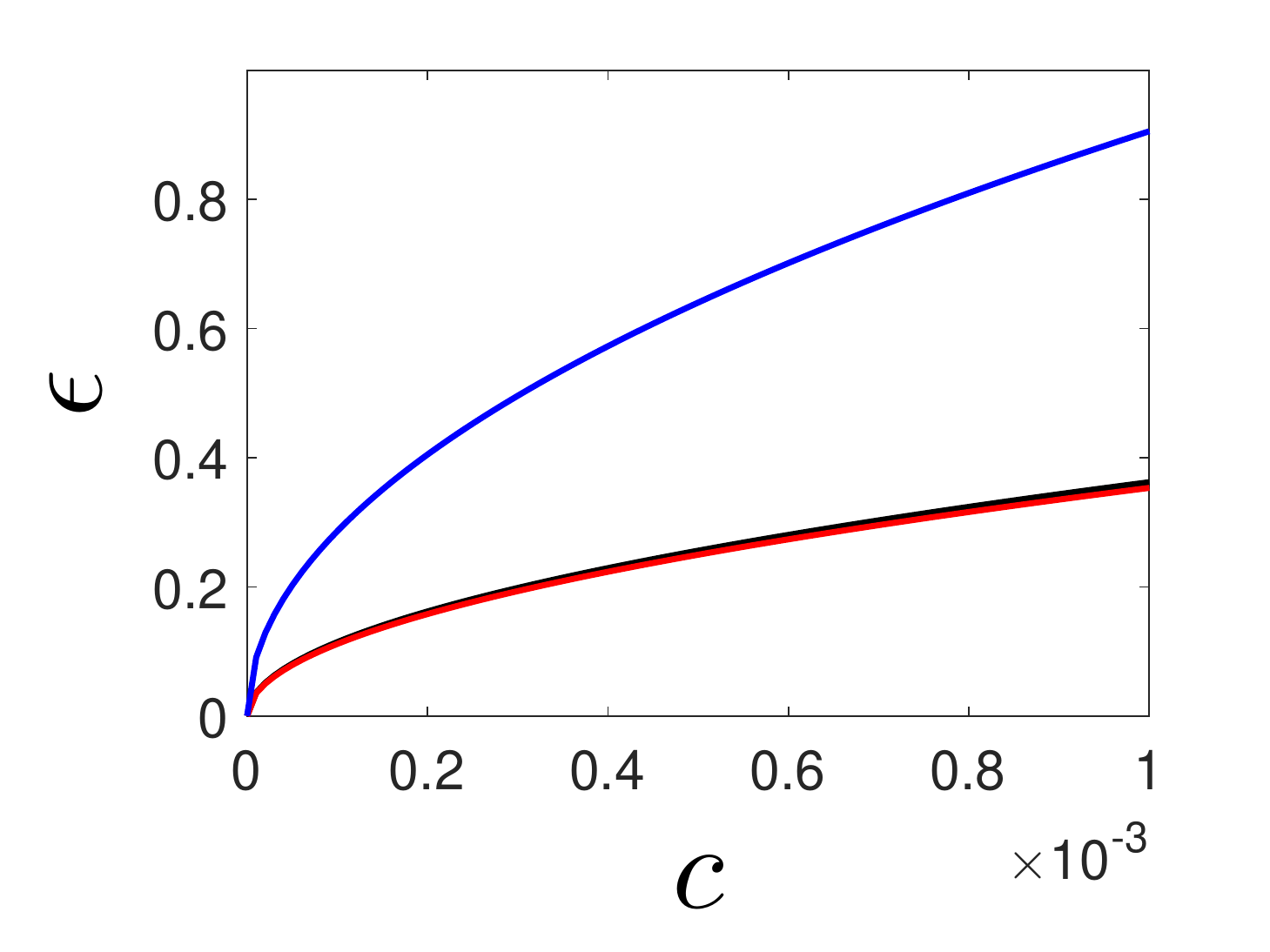}
\caption{(Color online) Evolution of $\epsilon$ (black and red lines) and $\epsilon_0$ (blue line) as a function of $c$ in the MFE case. The black and red curves depict respectively the numerical solution and the approximate expression of $\epsilon$ given by Eq.~(\ref{epsi1}). Parameters are set to $a=12.65$ and $b=0.0012$. Note that $ab< 0.079$.}
	\label{fig9}
\end{figure}

\begin{figure}[!ht]
	\centering
	\includegraphics[width=0.7\linewidth]{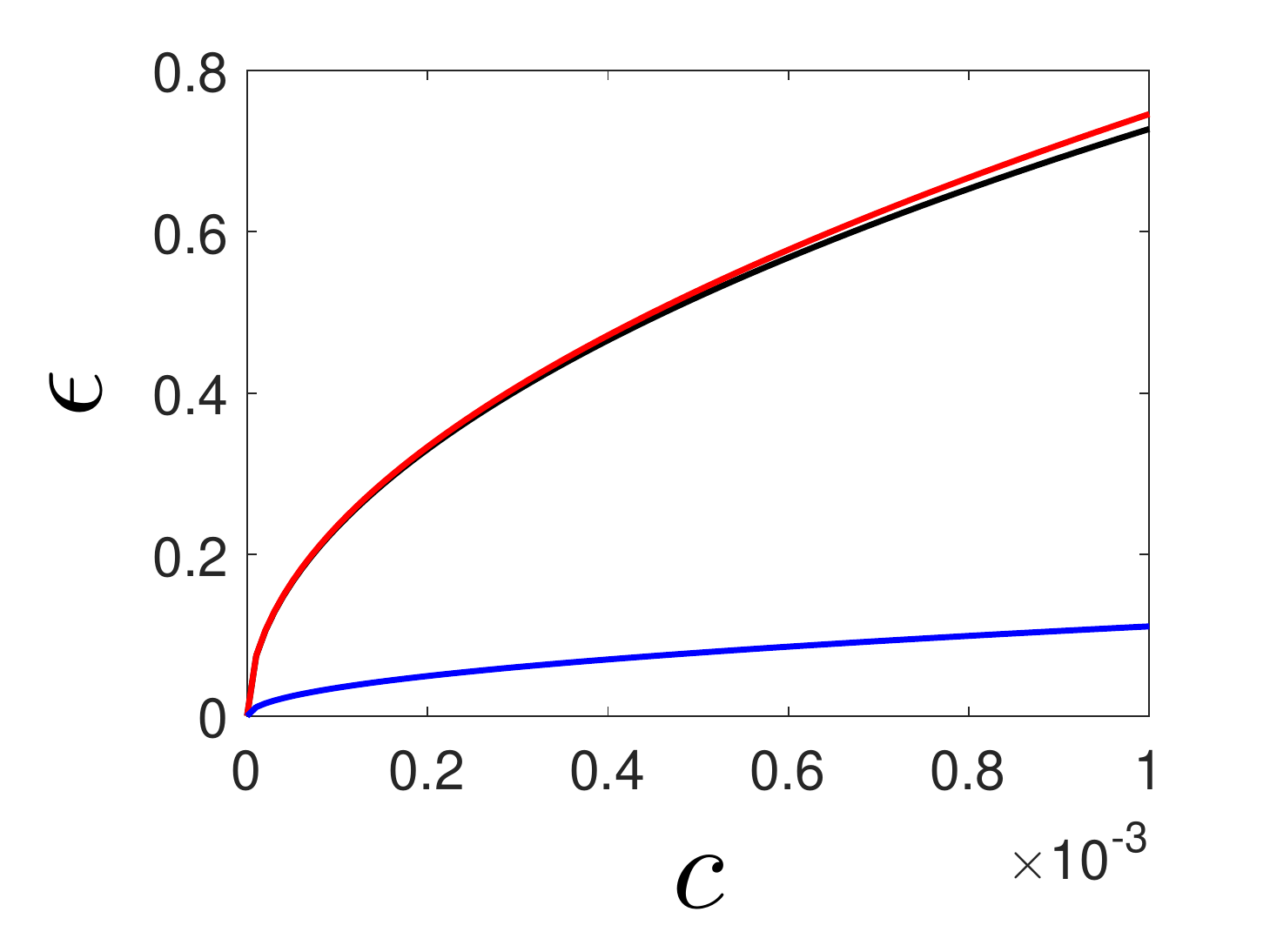}
\caption{(Color online) Same as Fig.~\ref{fig9} but for the region $\epsilon>\epsilon_0$. The approximate expression of $\epsilon$ is given by Eq.~(\ref{epsi2}). Parameters are set to $a=8.82$ and $b=0.0078$.}
	\label{fig10}
\end{figure}

\end{document}